\documentclass{sig-alternate-05-2015}

\usepackage{graphicx}
\usepackage[noadjust]{cite}
\usepackage{multirow}
\usepackage{amssymb}
\usepackage{graphicx}
\usepackage{amsmath}
\usepackage[tight,footnotesize]{subfigure}
\usepackage[linesnumbered,ruled]{algorithm2e}
\usepackage{algpseudocode}
\usepackage{epstopdf}
\usepackage{url}
\begin{document}
%\pagenumbering{arabic}

\title{Profitable Task Allocation in Mobile Cloud Computing}

\numberofauthors{3} %  in this sample file, there are a *total*
\author{
% 1st. author
\alignauthor
Mojgan Khaledi\\
%Ben Trovato\titlenote{Dr.~Trovato insisted his name be first.}\\
      \affaddr{School of Computing}\\
       \affaddr{University of Utah}\\
%       \affaddr{Wallamaloo, New Zealand}\\
      \email{mojgankh@cs.utah.edu}
% 2nd. author
\alignauthor
Mehrdad Khaledi\\
%Ben Trovato\titlenote{Dr.~Trovato insisted his name be first.}\\
      \affaddr{ECSE}\\
       \affaddr{ RPI}\\
%       \affaddr{Wallamaloo, New Zealand}\\
      \email{khalem@rpi.edu}
\alignauthor
Sneha Kumar Kasera\\
%Ben Trovato\titlenote{Dr.~Trovato insisted his name be first.}\\
      \affaddr{School of Computing}\\
       \affaddr{University of Utah}\\
%       \affaddr{Wallamaloo, New Zealand}\\
      \email{kasera@cs.utah.edu}}
\maketitle

\begin{abstract}

We propose a game theoretic framework for task allocation in mobile cloud computing that corresponds to offloading of compute tasks to a group of nearby mobile devices. Specifically, in our framework, a distributor node holds a multidimensional auction for allocating the tasks of a job among nearby mobile nodes based on their computational capabilities and also the cost of computation at these nodes, with the goal of reducing the overall job completion time. Our proposed auction also has the desired incentive compatibility property that ensures that mobile devices truthfully reveal their capabilities and costs and that those devices benefit from the task allocation. To deal with node mobility, we perform multiple auctions over adaptive time intervals. We develop a heuristic approach to dynamically find the best time intervals between auctions to minimize unnecessary auctions and the accompanying overheads. We evaluate our framework and methods using both real world and synthetic mobility traces. Our evaluation results show that our game theoretic framework improves the job completion time by a factor of 2-5 in comparison to the time taken for executing the job locally, while minimizing the number of auctions and the accompanying overheads. Our approach is also profitable for the nearby nodes that execute the distributor's tasks with these nodes receiving a compensation higher than their actual costs.

%Mobile devices (e.g. smart phones and tablets) have become the most effective tools of communication in our daily lives and they are expected to run computationally intensive applications. Despite the advances in technology, mobile devices still have limited resources. Therefore, a mobile device may need to borrow resources from its nearby devices for its computational job. In this paper, we present a game theoretic framework for task allocation in mobile cloud computing with selfish nearby devices. We propose a multidimensional auction that considers both the nearby nodes' cost for cooperation and the computational capacities they are willing to offer. The proposed framework provides incentives for selfish nearby devices to act altruistically, and it has desired economic properties that are formally proven in the analysis. Also, the proposed mechanism takes into account the heterogeneity of nearby resources and allocates the tasks in proportion to computational capacities. In addition, we consider mobility of devices in our framework by performing multiple auctions with dynamic time intervals. A heuristic approach is presented to find the time intervals between auctions. We evaluate our framework using both real world mobility traces and simulated mobility models, and the results show that the proposed mechanism improves the job completion time by factor of 2-5 compared to the case of executing the job locally, without increasing the communication overhead.
\end{abstract}

%\begin{IEEEkeywords}
%Mobile cloud computing, Game theory, Task allocation
%\end{IEEEkeywords}

% \fi
%\IEEEpeerreviewmaketitle

\section{Introduction}\label{secintro}

Personal mobile devices such as smartphones and tablets are increasingly being used for in our daily lives. Many more smartphones are being sold worldwide than the total sales of PCs~\cite{idc} and this growth in sales of smartphones is expected to continue in the future. Furthermore, various advances in technology are making these devices powerful tools capable of performing complex tasks, including speech to text conversion, audio identification, and image recognition.
However, despite these advances, resources on mobile devices are constrained by weight and size requirements of the device that must be met for the devices to be easily carried around. Therefore, mobile devices still have limited battery, storage, heat dissipation ability, etc which impede complex and resource intensive task execution. A possible remedy for tackling resource limitations of mobile devices is to offload the computational tasks to the cloud. 

%A possible remedy for tackling resource limitations of mobile devices when performing complex tasks is to offload computing tasks to nearby mobile devices. For example, in the speech to text conversion application, a mobile device can divide the audio file into smaller pieces, then assign each piece to a nearby device, and finally, combine the results obtained from nearby devices. 

Mobile Cloud Computing (MCC) generally refers to a client-server communication model where a mobile device (client) offloads computing tasks to the remote cloud through a wireless network (mainly cellular or WiFi networks)~\cite{cloud1, cloud2}. However, this model of mobile cloud computing is facing some important challenges. First, the performance of MCC is highly depend on wireless communication networks. With the tremendous growth of mobile data users, the wide-area mobile data access links (e.g., backhaul links in a cellular network) are becoming a bottleneck. This trend is expected to continue at even higher scales in the future because service providers, specifically cellular service providers, are unable to upgrade their backhaul networks due to shrinking profits. Thus, the MCC model can suffer from high latency and slow data transfer which may not be acceptable for the users of mobile applications. Second, although the cloud provides shared resources and amortizes costs, its operation requires establishment and maintenance of highly expensive hardware to run the high computational tasks. 

A possible remedy for tackling resource limitations of mobile devices when performing complex tasks, while not requiring the use of a server cloud and also minimizing latencies, is to offload computing tasks to nearby mobile devices. For example, in the speech to text conversion application, a mobile device can divide the audio file into smaller pieces, then assign each piece to a nearby device, and finally, combine the results obtained from nearby devices. Offloading compute tasks to nearby mobile devices rather than using a remote cloud through the mobile data cellular network lowers the latency and the burden on network backhaul. The nearby mobile devices, collectively and opportunistically, essentially provide the power of a cloud. Thus, we have another notion of mobile cloud computing that corresponds to offloading of compute tasks to a group of nearby mobile devices connected by various types of links including D2D, WiFi Direct, Bluetooth, etc. In this paper, we use this second notion of mobile cloud computing. 

%Cloudlets~\cite{cloudlet} and fog computing~\cite{fogcomputing} tried to address these issues by placing more computation resources near the edge to lower the latency and the burden on backbone network. However, this placement of computation resources can not handle the privacy and the cost issues. In addition, large scale deployments of cloudlet and fog computing remains a significant challenge.

%Personal mobile devices such as smart phones and tablets are increasingly being used for communication in our daily lives. According to the International Data Corporation (IDC)s statistics~\cite{idc}, smartphones were sold worldwide more than three times of the total sales of PCs and this growth in sales will expected to be continued in future. Furthermore, various advances in technology are making these devices powerful tools capable of performing complex tasks, including speech to text conversion, audio identification, and image recognition. According to recent studies, now a top-of-the-line smartphone can do about 1.5 billion numerical operations a second which is equal to 25 percent of a modern PC computing power. This 25 percent may go to 100 percent in future. With proliferation and increasing the computation power of personal mobile devices, researchers found a new source for computing. Unlike PCs, personal mobile devices are available even when they are charging or not in use. Also, personal mobile device users only utilize one percent of computational capability of their mobile devices. Thus, there is a tremendous inexpensive amount of computing power around us.

There is a growing amount of work to utilize mobile device computing power for cloud computing. Hyrax~\cite{hyrax} uses the computational power of a network of android smartphones in MapReduce. Mobile Device Clouds~\cite{fahim1,fahim2} and Serendipity~\cite{mostafaamar} are platforms for opportunistic computing where a mobile device offloads computing tasks to nearby mobile devices. NativeBOINC for Android~\cite{mobileboinc} is another examples of utilizing mobile devices' computing power. Recently, Habak et al.~\cite{femtocloud} proposed FemtoCloud where a controller executes a variety of tasks arriving at controller by using the computational power of nearby mobile devices. SymbIoT~\cite{vision} is another platform that uses the computational capability of all mobile devices within the same network to perform different tasks. 

%These efforts have led to another notion of MCC that corresponds to offloading of compute tasks to a group of nearby mobile devices connected by various types of links including D2D, WiFi Direct, and Bluetooth.%~\cite{hyrax, mostafaamar, Mcloudcmu}. 

%Offloading compute tasks to nearby mobile devices rather than using a remote cloud through the mobile data cellular network lowers the latency and the burden on network backbone. It also preserves privacy by processing the data locally and not send it over network. Moreover, using the computational power of mobile devices that already exist in the network instead of deploying cloud servers reduces the cost. 

%serious effort is being made towards artificially enhancing the mobilthe data access network capacity through localized device-to-device (D2D)~\cite{hyrax, mostafaamar, Mcloudcmu} and WiFi Direct~\cite{microcast} communication wherever possible, as well as offloading compute tasks to nearby nodes rather than using a remote cloud through the mobile data cellular network. This and similar efforts have led to another notion of MCC that corresponds to offloading of compute tasks to a group of nearby mobile devices connected by various types of links including D2D, WiFi Direct, and Bluetooth~\cite{hyrax, mostafaamar, Mcloudcmu}. Recently, Fahim et al.~\cite{Mcloudcmu} presented a platform for offloading of compute tasks to nearby mobile devices and showed that D2D offloading performs better than remote cloud offloading in terms of time and energy especially for computationally intensive tasks. 
In this paper, we 
%consider this notion of MCC and 
examine a scenario where a mobile device or a central controller, that we call a \textit{distributor} node, has a computational job or a set of different computational jobs and wants to utilize resources of nearby mobile nodes to reduce the job completion time. %For example, in the speech to text conversion application~\cite{audio2text}, a mobile device can divide the audio file into smaller pieces, then assign each piece to a nearby device, and finally, combine the results obtained from nearby devices. %, or in SymbIoT~\cite{vision} or FemtoCloud~\cite{femtocloud} platforms a central controller with a set of different jobs that comes from different applications uses the nearby mobile devices computing power to reduce the total job completion time for all applications. 
Due to mobility, the distributor (mobile device or a central controller) has frequent contacts with other mobile devices that can provide the required computational resources. The problem faced by the distributor is how to select the nearby nodes and divide the job among them in a manner that is beneficial to all the parties involved.

While the problem of using other nodes for executing the tasks of a job has been widely used in distributed computing~\cite{setihome ,boinc }, cyber foraging~\cite{cyberforaging1}%, cyberforaging2}
, and {\em crowdsourcing}~\cite{crowdsourcing1},%, crowdsourcing2}, 
%we argue that 
the existing work on task allocation cannot be simply adopted for mobile cloud computing. First, the task allocation method must take into account the selfish behavior of mobile nodes by providing incentives for them. This because in mobile cloud computing, a rational mobile node would not be willing to lend its resources (and thereby deplete its battery) unless it receives some payoff as compensation. Second, mobile devices can have different hardware/software and thus have different capabilities. For this reason, the execution time of a specific task can be different across mobile devices. Therefore, the task allocation mechanism must consider the heterogeneity of mobile devices to reduce the overall job completion time. Third, the task allocation needs to take into account the mobility of nodes. In a mobile environment, the distributor may observe disconnection of nodes with assigned tasks, and new arrivals that might provide high computational capabilities. Thus, decisions should be made according to the dynamics of the environment. Fourth, and very importantly, the distributor node should see a clear benefit in terms of job completion time. 

%It is worth noting that simply dividing the job into %several tasks of the same workload and allocating them to %nearby devices does not necessarily yield minimum job %completion time. 

We propose a game theoretic framework for task allocation that provides incentive for all mobile nodes. In our framework, the distributor node holds a multidimensional auction for allocating the tasks of a job among nearby mobile nodes based on their computational capabilities and also the cost of computation at these nodes, with the goal of reducing the overall job completion time. To the best of our knowledge, this is the first work that presents a multidimensional auction for task allocation in mobile cloud computing. Our proposed auction also has desired economic properties (that we formally prove later in the paper) including {\em incentive compatibility} which ensures that players truthfully reveal their capabilities and costs, and that mobile nodes act cooperatively in the proposed auction for the benefit of all the parties involved. We also consider the mobility of mobile nodes in our game theoretic framework. In such a mobile environment, the topology of the network may change over time. %~\cite{khaledimobility1}. 
Thus, some nearby nodes may get disconnected from the distributor before completing the assigned task, resulting in an increased job completion time. To deal with mobility, we perform multiple auctions over adaptive time intervals. We develop a heuristic approach to dynamically find the best time intervals between auctions to minimize unnecessary auctions and the accompanying overheads. We briefly explore the privacy of the distributor and the nearby mobile devices and show a tradeoff between providing privacy, and the profits for the parties involved. 

We evaluate our framework and methods using both real world and synthetic mobility traces. We use two models of compute jobs - a simple single job model, and a multiple job model that uses a Directed Acyclic graph to  represent causal dependencies in a set of jobs. Our evaluation results show that our game theoretic framework improves the job completion time by a factor of 2-5 in comparison to the local execution of the job, in both the job models, while minimizing the number of auctions. Thus, our approach is beneficial for the distributor in terms of enhancing its performance. We also show that the nearby nodes that execute the distributor's tasks receive a compensation higher than their actual costs. %(Section~\ref{secsysappoach}).

\section{Related work}\label{secrelwork}

Many existing works in distributed computing (e.g.,SETI@
Home~\cite{setihome}, BOINC~\cite{boinc}, and cyber foraging~\cite{cyberforaging1}%, cyberforaging2}
) have proposed using other nodes for executing the tasks of a job. 
%More recently, crowdsourcing has been proposed for obtaining services and content by soliciting contributions from a large group of people/devices~\cite{crowdsourcing1, crowdsourcing2}.
However, all of these existing works primarily assume altruistic behavior in the distributed computing environment and do not carefully incentivize resource sharing. Like our work, Serendipity~\cite{mostafaamar} enables remote computing among a set of intermittently connected mobile devices. However, our work differs from Serendipity in the following significant ways. First and foremost, Serendipity does not incentivize resource sharing. Second, it does not consider heterogeneity among mobile devices in task allocation. Moreover, all assigned tasks are assumed to have equal workload. Third, Serendipity does not use any adaptive methods for reassigning tasks. 

Our game theoretic framework is inspired by the multidimensional mechanism proposed for the second score auction~\cite{multidimensionalscoreauction}, where the authors use a linear function to map the multidimensional bid into a single dimension. However, we use a fractional function for the mapping which is more suitable to our setting (see equation~\ref{equScore}). We also extend the existing multidimensional mechanism by allowing selection of $k$ items while considering a budget limit instead of selection of only one item. The use of multidimensional auction allows us to consider the heterogeneity of nearby resources in task allocation. We also minimize the number of auctions and the accompanying overhead by developing a heuristic approach for finding the best time intervals between auctions.

Existing incentive mechanisms that have been used for task allocation~\cite{incetivetaskallocation1,incetivetaskallocation2, incentivecrowdsourcing}%,incentivecrowdsourcing2}
, do not consider multidimensional auction where both cost and the service quotient are important. The authors in~\cite{khaledioptimalauction} proposed a multidimensional optimal auction to provide incentive in mobile ad hoc network routing, considering both cost and the path duration in route selection. However, their proposed approach only works for time related bids such as path duration where the players cannot over-report the time related bid. In our task allocation problem, players can over-report and under-report both cost and the committed service quotient.

\section{Problem Formulation And Solution}\label{secsysappoach}

%We model the task allocation in mobile cloud computing as a multidimensional auction. The proposed auction implements dominant equilibrium where players select their strategies regardless of other players' strategies. The distributor, the smart phone who wants to offloads tasks to its nearby nodes, holds the auction and determines the selected nodes for task allocation.

%First, the distributor sends a probe message to find the nearby nodes, and the nearby nodes response with information about the cost and committed quality which determines the amount of data that they can process per second. For example, for the speech to text conversion, the committed quality is the amount of audio file that the nearby nodes can convert in one second. The mobile node finds this information based on its processor power, storage, etc. 
\begin{table}[!t]
\caption{Notation}
\label{tblnotation}
\centering
\resizebox{8cm}{!}{\begin{tabular}{c c }
\hline

 Parameter & Definition \\
\hline
\hline
$(c_i, s_i)$& The cost and committed service quotient for player $i$ \\
$a_{c_i, s_i}$& Allocation rule for player $i$ with type $(c_i, s_i)$\\
$u_i(c_i, s_i)$& Utility of player $i$ with type $(c_i, s_i)$\\
$p(a_{c_i, s_i})$& Payment to player $i$ under allocation rule $a_{c_i, s_i}$\\
$B$& The distributor budget limit\\
$D$& The total workload of job\\
$T$& Time interval between auctions\\
$b$& Linear increase factor in heuristic approach\\
\hline
\end{tabular}}
\end{table}

We consider mobile cloud computing in the presence of selfish smart phones. We assume that all smart phones act rationally and selfishly, and their main goal is to maximize their own profits, not to harm others. There is a distributor node that wants to offload a computational job, with a total work load of $D$ units, to its nearby nodes with the goal of reducing the overall job completion time. Since a smart phone incurs a cost (in form of resource and battery usage) while performing a task on behalf of the distributor, it may not be willing to participate unless the distributor provides right incentives.

We propose a multidimensional auction to model the job/
task allocation in mobile cloud computing. In our auction, the distributor holds the auction among $n$ nearby smart phones called \textit{players}. Each player $i$ has an individual private value $t_i$ called its \textit{type} which consists of the following two parameters:

\begin{itemize}

\item $s_i$: the committed service quotient that player $i$ can provide. In distributed computing applications such as speech to text conversion, the service quotient is related to the amount of data that the smart phone can process in a given time (the smart phone finds this information by estimating its maximum execution time running synthetic benchmark and its energy consumption using techniques like PowerBooter~\cite{powerbooter}). In our system, $s_i$ denotes the amount of data that player $i$ can process per second. It should be noted that the distributor is not aware of the actual $s_i$s until the tasks are completed and the results are received from nearby players who have been assigned the tasks.

\item $c_i$: the cost of player $i$ for performing the task with the committed service quotient. $c_i$ is a function of the committed service quotient, $s_i$, and the player's private cost of performing the task, $\theta_i$. We bound $\theta_i$ such that $0<\theta_{min}\leq\theta_{max}<\infty$. $\theta_i$ is affected by various parameters such as processor speed, available storage, remaining battery level, communication cost, etc. $c_i$ is an increasing function of both $s_i$ and $\theta_i$, and is private information of the player $i$, hence, no one else can determine the exact value of $c_i$.
%Also, $\theta_i \in [\theta_{min}, \theta_{max}]$ and %$0<\theta_{min}\leq\theta_{max}<\infty$. %We model %$c_i(b_i, \theta_i)=s_i\theta_i$.

\end{itemize}

Our mechanism works as follows. First, the distributor sends a probe message to find the nearby nodes. Then, each nearby node $i$, interested in participating, replies by announcing its bid $(c_i, s_i)$. Note that $(c_i, s_i)$ announced by the node $i$ need not be the actual private value of its type, $t_i$. Based on the received bids, the distributor selects a set of players and assigns the tasks to them in proportion to their committed service quotient. To provide incentive to the players for their resources, the distributor compensates them by paying the players. The distributor also has incentive to offload tasks to nearby mobile devices because the overall job completion time is less than the time taken for executing the job locally. Therefore, our game theoretical framework is beneficial to all parties involved.

Our mechanism implements truthfulness in \textit{dominant equilibrium} implying that each player's best strategy, regardless of other players strategies, is to bid truthfully (i.e. to report actual values of $s_i$ and $c_i$). Table~\ref{tblnotation} summarizes the notation we use in this and the subsequent sections.

\subsection{Allocation Mechanism}
In this section, we specify how the distributor selects the nearby mobile devices by considering both the cost and the service quotient. The distributor wants players with minimum cost and maximum service quotient. For this reason, the distributor defines a weights function for player $i$ with type $(c_i, s_i)$ as follows:
\begin{equation}\label{equScore}
w_i=\frac{s_i}{c_i}
\end{equation}
The distributor must also determine how many players it should select. Intuitively, the number of players affects the job completion time and the sum of premiums paid by the distributor to the players over the players' actual costs (\textit{overpayments}). As the number of selected players is increased, the job completion time decreases due to more tasks being executed in parallel. However, the increase in number of selected players also increases the overpayments.

To limit its cost, the distributor defines a budget limit $B$ as the maximum amount that it can pay for the processing of one unit of data per second. $B$ is a function of the distributor utility, $u_d$. An increase in the distributor utility, $u_d$, makes the distributor willing to increase its budget limit. Formally,
\begin{equation}\label{eqbudgetimit}
B(u_d)=\left\{
\begin{array}{@{} l c @{}}
f^+(u_d^{max}) & u_d > u_d^{max} \\
f^+(u_d)& u_d^{min}\leq u_d \leq u_d^{max}\\
0 &u_d<u_d^{min}
\end{array}\right.
\end{equation}
Here, $u_d^{min}$ and $u_d^{max}$ are the minimum and the maximum utilities that the distributor expects when using our task allocation framework. When the obtained utility, $u_d$, is less than $u_d^{min}$, the distributor prefers to execute the job locally. By increasing the value of $u_d$ between $u_d^{min}$ and $u_d^{max}$, the distributor is willing to use the task allocation framework. In this case, the budget limit is a nondecreasing function of $u_d$ ($f^+(u_d)$ represents the nondecreasing function of $u_d$). $u_d^{max}$ is the saturation point. Increasing the value of $u_d$ beyond this saturation point does not increase the budget limit, $B$. The distributor determines $u_d^{min}$, $u_d^{max}$, and $f^+(u_d)$ depending on the application. For example, the distributor may prefer to choose a constant function for $f^+(u_d)$, if the distributor utility, $u_d$, be above $u_d^{min}$ or the distributor may select an increasing exponential function for $f^+(u_d)$ in real time applications where time is critical.

Let $t_O$ represent the job completion time when the distributor uses our task allocation framework, and let $t_L$ represent the job completion time when the distributor executes the job locally. Also, let $e_O$ and $e_L$ be the energy consumption when the distributor using our task allocation framework and when the distributor executing the job locally. Then, the distributor utility, $u_d$, will be:
\begin{equation}\label{equutilitydistributor}
u_d=\alpha(t_L-t_O)+\beta(e_L-e_O)
\end{equation}

The values of $\alpha$ and $\beta$ determine the importance of job completion time and energy consumption in the distributor utility. In our evaluation, the distributor only wants to reduce the job completion time and ignores the energy consumption (we set $\alpha=1$ and $\beta=0$). We ignore the energy consumption for the following reasons. First, our task allocation is designed for high computational tasks such as audio identification and image recognition. In this case, the energy consumption for executing task, $e_L$, is greater than the energy consumption for communication and transferring data among nearby mobile nodes, $e_O$. Second, in our task allocation framework, a computing cloud refers to a group of nearby mobile devices that connect by WiFi and Bluetooth. Thus, the energy consumption for communication and data transferring is less than the common client-server cloud. 

Also, note that the distributor computes $t_O$ by considering the pre-processing time, the overhead of dividing job into tasks, the post-processing overhead of assembling the results, and the maximum job completion time among the selected nearby mobile nodes (the distributor can estimate the job completion time in each nearby mobile user by using their declared service quotients).

Now the goal is to select a subset of players that maximize the total weight under the budget limit constraint. We use a simple and efficient greedy approach for allocation. First, the distributor orders the players based on their decreasing weights. Next, it selects the largest number of players $\lbrace 1, 2, 3, \ldots, k \rbrace$ that satisfy the budget limit constraint (i.e. constraint~\ref{bg2}). By considering the budget limit constraint, the distributor obtains at least the minimum utility, $u_d^{min}$. Therefore, participating in the task allocation is beneficial for the distributor. In the next section, we explain the payment function that the distributor pays to the nearby mobile users to provide incentives for them to participate in the task allocation.
\subsection{Payment Mechanism}
In this section, we determine the payment Mechanism. The payment Mechanism must provide the following desirable economic properties to ensure that players act cooperatively and bid truthfully.
\begin{itemize}
\item \textit{Individual Rationality}: The utility of all players should always be non-negative. Otherwise, players may choose to  not participate.
\item \textit{Incentive Compatibility}: In an incentive compatible mechanism, no selfish node has incentive to lie (also called truthfulness).
\end{itemize}
Let $a_{c_i, s_i}\in\lbrace 0, 1\rbrace$ denote the allocation to player $i$ with type $(c_i, s_i)$ where if the distributor offloads task to the player, $a_{c_i, s_i}=1$, otherwise, $a_{c_i, s_i}=0$. Also, let $p_i(a_{c_i, s_i})$ be the payment that the distributor pays to the player $i$ under allocation rule $a_{c_i, s_i}$. Then, the utility of player $i$ with type $(c_i, s_i)$ is obtained from the following formula:
\begin{equation}\label{equutility}
u_i(c_i, s_i)=p_i(a_{c_i, s_i})- c_i
\end{equation}
Having determined the utility of players, the payment Mechanism must satisfy the following constraints.  
%\max_a \sum_{i=1}^n w_i(c_i, s_i) a_{c_i, s_i}\nonumber\\
%s.t.\nonumber\\
\begin{gather}
\forall i, c'_i, p_i(a_{c_i, s_i})-c_i\geq  p_i(a_{c'_i, s_i})-c_i\label{equattIC1}\\
\forall i, s'_i, p_i(a_{c_i, s_i})-c_i\geq  p_i(a_{c_i, s'_i})-c_i\label{equattIC2}\\
\forall i, c'_i, s'_i, p_i(a_{c_i, s_i})-c_i\geq  p_i(a_{c'_i, s'_i})-c_i\label{equattIC3}\\
\forall i, p_i(a_{c_i, s_i})-c_i \geq 0 \label{ir}\\
\sum_{i=1}^n \frac{p_i(a_{c_i, s_i})}{s_i}\leq B(u_d)
 \label{bg2}
\end{gather}
%We observe in the above optimization problem that allocation rule should maximize the total weight. Recall that $w_i=\frac{s_i}{c_i}$. Thus, we consider both cost and the committed service quotient in our optimization problem. 
Constraints \ref{equattIC1}, \ref{equattIC2}, and \ref{equattIC3} provide incentive compatibility for both cost and the committed service quotient. Constraint \ref{ir} is for individual rationality, and constraint \ref{bg2} captures the budget limit of the distributor.

The payment to player $i$ consists of two parts:
\begin{equation}\label{eqpayment}
p_i(a_{c_i, s_i})=\left\{
\begin{array}{@{} l c @{}}
p^1_i(a_{c_i, s_i})+p^2_i(a_{c_i, s_i}) & i\leq k\\
0 & i > k
\end{array}\right.
\end{equation}
Where $p^1_i(a_{c_i, s_i})$ is paid to provide incentive compatibility for $c_i$ and $p^2_i(a_{c_i, s_i})$ is paid to provide incentive compatibility for both $c_i$ and $s_i$. $p^1_i$  is obtained from the following formula:
\begin{equation}\label{eqpay1}
p^1_i(a_{c_i, s_i})= \frac{s_i c_{k+1}}{s_{k+1}}
\end{equation}

Here, $k+1$ is the index of the player with largest weight after the selected $k$ players.

Let $s_i$ be the actual service quotient (e.g., the amount of data that player $i$ can process per second) that the distributor finds after the task completion by player $i$, and also let $s'_i$ be the announced service quotient by the same player. Then, the second payment, $p^2_i$, is of the following form.
\begin{equation}\label{eqpay2}
p^2_i(a_{c_i, s_i})=d(s_i, s'_i)(s_i-s'_i)
\end{equation}

Here, $d(s_i, s'_i)$ is a positive and nondecreasing function of $(s_i-s'_i)<0$, that determines the impact of misreporting the service quotient. A suitable choice for the function $d(s_i, s'_i)$ must satisfy the following three properties. First, if the player $i$ over-reports its committed service quotient ($s_i < s'_i$), then $p^2_i(a_{c_i, s_i})$ must be negative, i.e., the player must give back some money to the distributor. Second, a player must not pay a penalty for under-reporting its committed service quotient, i.e., when $s_i > s'_i$, $d(s_i, s'_i)=0$. Third, the choice of $d(s_i, s'_i)$ should lend itself to satisfying the incentive compatibility property. In this paper, we define $d(s_i, s'_i)$ as follows:
\begin{equation}\label{eqdefic}
d(s_i, s'_i)=\left\{
\begin{array}{@{} l c @{}}
0 & s_i\geq s'_i\\
(s'_i-s_i)^2+\frac{c_{k+1}}{s_{k+1}} & s_i < s'_i
\end{array}\right.
\end{equation}

\begin{figure}[!t]
\centerline{
\subfigure[]{\label{fig_d}\includegraphics[width=1.8in,height=1in]{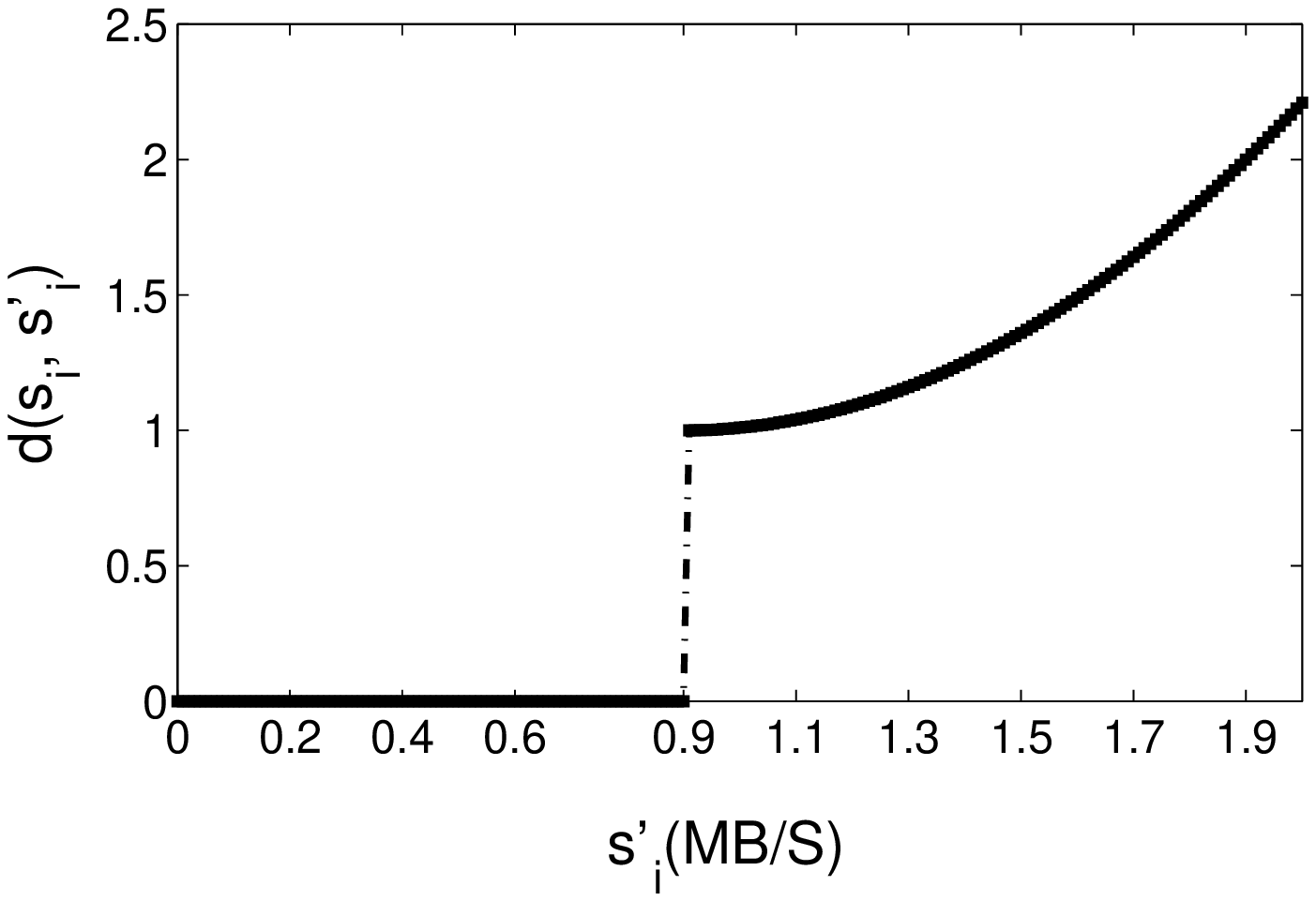}}
\hfil
\subfigure[]{\label{fig_utility}\includegraphics[width=1.8in,height=1in]{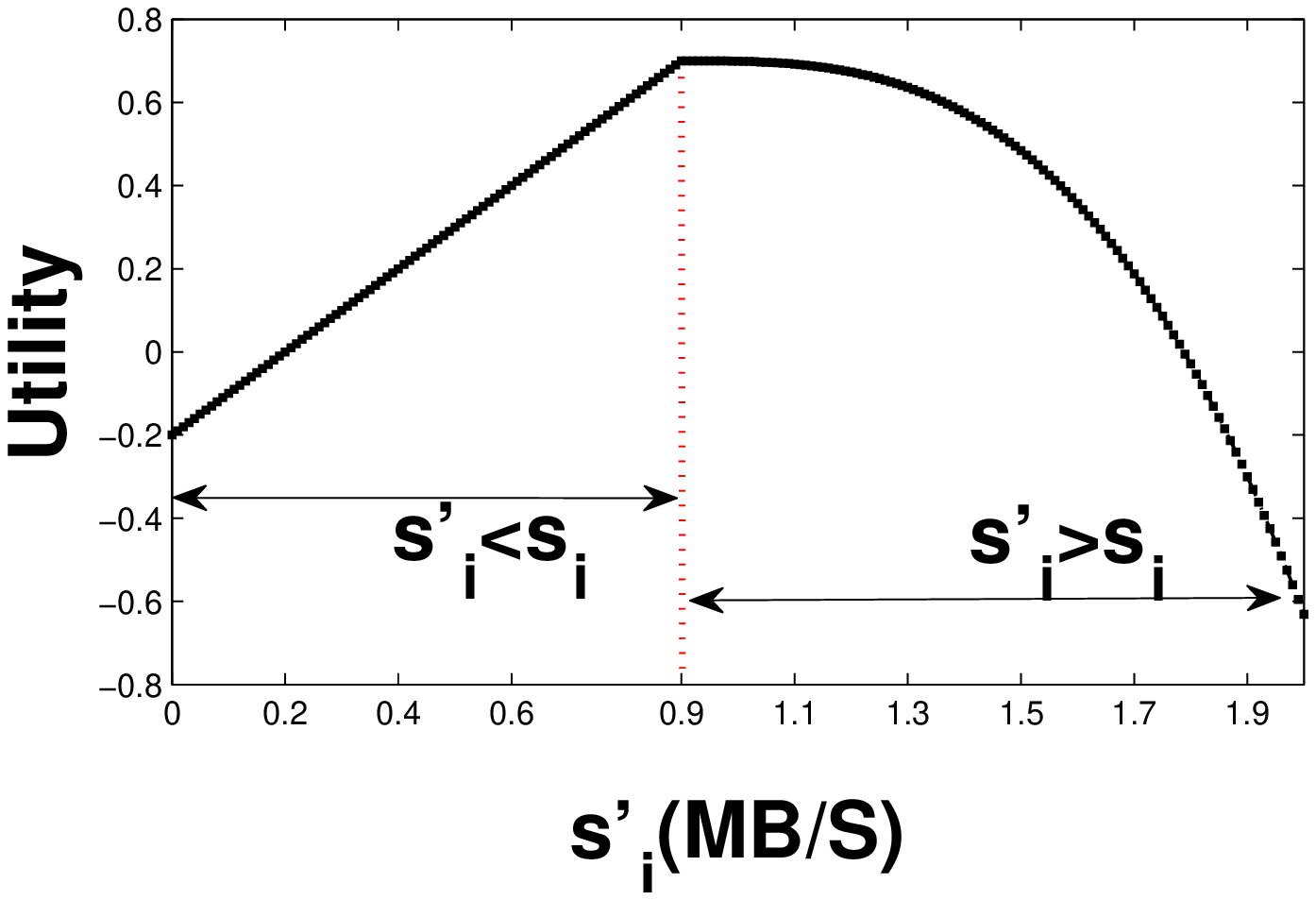}}}
\caption{$d(s_i, s'_i)$ (a) and the utility (b) for different values of $s'_i$, with $s_i=0.9$ Mbytes/sec.}
\label{fig_dutility}
\end{figure}
This definition of $d(s_i, s'_i)$ satisfies the first two required properties described above. We show that this choice of $d(s_i, s'_i)$ helps satisfy the incentive compatibility property in Lemma~\ref{lem2}. %Note that other choices of $d(s_i, s'_i)$ are also possible but we do not consider those in this paper. 
 Figure \ref{fig_d} shows an instantiation of $d(s_i, s'_i)$. 
%In this figure, $s$ is the portion of the audio file that a player can process per second. 
In this figure, the actual value of the service quotient is $0.9$ Mbytes/sec. As shown in the figure, the value of $d(s_i, s'_i)$ increases if the player declares its committed service quotient ($s'_i$) to be greater than its actual service quotient. Also, the value of $d(s_i, s'_i)$ is zero if the player under-reports its committed service quotient.

Note that other choices of $d(s_i, s'_i)$ are also possible. Depending on the application, we can choose different functions for $d(s_i,s'_i)$. For example, in some applications the value of $s_i$ might vary with changes in the environment that are not within the of control of the player. To provide incentive compatibility and to also prevent nonessential punishment, we can also define $d(s_i, s'_i)$ as follows:
\begin{equation}\label{eqdefic1}
d(s_i, s'_i)=\left\{
\begin{array}{@{} l c @{}}
0 & s_i\geq s'_i\\
\frac{c_{k+1}}{s_{k+1}} & s_i < s'_i
\end{array}\right.
\end{equation}
In this case, if the player over-reports its committed service quotient ($s_i < s'_i$), then the payment is:
\begin{align*}
p_i(a_{c_i, s'_i})=&\frac{s'_i c_{k+1}}{s_{k+1}}+(s_i-s'_i)\frac{c_{k+1}}{s_{k+1}}=\frac{ c_{k+1}}{s_{k+1}}s_i=p_i(a_{c_i, s_i})
\end{align*}
This payment is the same as the payment if the player wants to act truthfully. Therefore, the player has no incentive to lie. On the other hand, if the committed service quotient reduces after the player declares it, the player is paid only for the service quotient it accomplishes. We can also add a small reward if the player under-reports its committed service quotient ($s_i>s'_i$). The amount of reward should be less than the amount of payment when the player acts truthfully ($s_i=s'_i$). The player has no incentive to under-report its committed service quotient. However, if the the committed service quotient increases after the player declares it, the player has an added incentive to provide a better service quotient.

Having determined the allocation rule and the payment policy, we must now prove the economic properties of our multidimensional auction, namely individual rationality and incentive compatibility using the definition of $d(s_i, s'_i)$ in equation \ref{eqdefic}. The same reasoning can be applied to other definitions of $d(s_i, s'_i)$, e.g., in equation \ref{eqdefic1}. %Our auction mechanism implements these properties in a dominant equilibrium where each player's best strategy is to act truthfully, regardless of what other players do.   

\subsection{Proofs}\label{subsecproofs}
Players may cheats about their types, $(c, s)$, to gain extra profit. Incentive compatibility (IC) ensures that players truthfully reveal their actual types. In a multidimensional mechanism, proving IC is challenging. This is because several cheating scenarios, formed from combinations of cheating in each dimension, must be considered. We must prove IC in the following conditions.
\begin{enumerate}
\item The truthful revelation of a player's cost ($c$) is a dominant equilibrium, given that the player reveals its committed service quotient ($s$) truthfully (constraint \ref{equattIC1}).
\item The truthful revelation of a player's committed service quotient ($s$) is a dominant equilibrium, given that the player reveals its cost ($c$) truthfully (constraint \ref{equattIC2}). 
\item The truthful revelation of a player's cost and committed service quotient ($c, s$) is a dominant equilibrium (constraint \ref{equattIC3}).
\end{enumerate}

\newtheorem{lemma}{Lemma}
\begin{lemma}\label{lem1}
Given the committed service quotient ($s$), truthfully revealing the cost ($c$) results in a dominant equilibrium.
\end{lemma}
\begin{proof}
 There are two possible cases:
 \begin{itemize}
 \item
 \textit{Player is one of the winners.} If by over-reporting or under-reporting, the player still remains a winner, the utility of the player does not change. If the player become a loser by over-reporting, then the utility of the player becomes zero which is less than that it can obtain by acting truthfully. Therefore, in this case, the player has no incentive to lie.
 \item
  \textit{Player is one of the losers.} If by over-reporting or under-reporting the player remains a loser, the utility of the player dose not change (it is still zero). On the other hand, if a player, (without loss of generality) $k+1$, under-reports the value of $c_{k+1}$ to become a winner (i.e., $\frac{s_{k+1}}{c'_{k+1}} >\frac{s_{k}}{c_{k}}$), its utility becomes: 
 \begin{equation}
 u_{k+1}(c'_{k+1}, s_{k+1})=\frac{s_{k+1} c_k}{s_k}-c_{k+1}
 \end{equation}
 However, we know that $\frac{s_k}{c_k}> \frac{s_{k+1}}{c_{k+1}}$. This means that $c_{k+1}>\frac{s_{k+1} c_k}{s_k}$. Therefore, the utility is negative and the player has no incentive to lie.
 \end{itemize}
\end{proof}
\begin{lemma}\label{lem2}
Given the cost ($c$), truthfully revealing the committed service quotient ($s$) results in a dominant equilibrium.
\end{lemma}
\begin{proof}
We show that the utility of the player has its maximum value at $s'_i= s_i$. Thus, the player has no interest to misreport $s$.

By substituting $p^1$ and $p^2$ in equation \ref{equutility} and taking derivative with respect to $s'_i$,% the utility of the winner can be rewritten as follows.
%\begin{equation}\label{equtility2}
%u_i(c_i, s'_i)=\frac{s'_i c_{k+1}}{s_{k+1}}+d(s_i, s'_i)(s_i-s'_i)-c_i
%\end{equation}
%By taking derivative with respect to $s'_i$,
\begin{align*}
&\frac{\partial  u_i}{\partial s'_i}=\frac{c_{k+1}}{s_{k+1}}+\frac{\partial d(s_i, s'_i)}{\partial s'_i}(s_i-s'_i)-d(s_i, s'_i)
\end{align*}
For $s'_i=s_i$,  $\frac{\partial  u_i}{\partial s'_i}=0$. For $s'_i \in (s_i,+\infty ]$, given that $\frac{\partial d(s_i, s'_i)}{\partial s'_i}$ is positive, $\frac{\partial  u_i}{\partial s'_i}$ is negative. Also, for $s'_i \in (0,s_i]$, $\frac{\partial  u_i}{\partial s'_i}$ is positive. As a result, $s'_i=s_i$ is the maximum point. Figure~\ref{fig_utility}, shows the utility of a player for different values of $s'_i$. This figure also shows that the utility is maximum at point $s'_i=s_i$.
\end{proof}
%\begin{figure}[!t]
%\centering
%\includegraphics[width=2in,height=1in]{utility}
%\caption{The utility by varying the value of $s'_i$, with %$s_i=0.9$Mbytes/sec.}
%\label{fig_utility}
%\end{figure}

\begin{lemma}\label{lem3}
Truthfully revealing both cost and the committed service quotient $(c, s)$ results in a dominant equilibrium.
\end{lemma}
\begin{proof}
Let $c'$ and $s'$ be the declared cost and the committed service quotient of a player. Also, let $c$, and $s$ be the actual cost and preformed service quotient. Then, to prove incentive compatibility, we should consider the following four cases.
\begin{enumerate}
\item \textit{$c'< c$ and $s'< s$.}
If the player is already a winner and also wins by misreporting, then the utility of the winner will be less than that when acting truthfully.
\setlength\abovedisplayskip{0pt}
\begin{align*}
%u_i(c_i, s_i)=\frac{s_i c_{k+1}}{s_{k+1}}-c_i\\
u_i(c'_i, s'_i)=\frac{s'_i c_{k+1}}{s_{k+1}}-c_i
\end{align*}
\setlength\belowdisplayshortskip{0pt}Given that $s'_i<s_i$, $u_i(c_i, s_i)>u_i(c'_i, s'_i)$ and the player has no incentive to lie.
If the player is not a winner, but wins the game by misreporting, then the utility of the player becomes negative (proof is similar to the second case of Lemma \ref{lem1}).
\item  \textit{$c'> c$ and $s'< s$.}
If the player is already a winner and still wins the game by misreporting, the utility of the player will decrease. On the other hand, if the player is not a winner, then it cannot win by misreporting. Thus, its utility does not change.
\item  \textit{$c'< c$ and $s'> s$.}
If the player is already a winner and remains a winner by misreporting, then the utility of the player is as follows.
\setlength\abovedisplayskip{0pt}
\begin{equation*}
u_i(c'_i, s'_i)=\frac{s'_i c_{k+1}}{s_{k+1}}-c_i+(s_i-s'_i)(\frac{c_{k+1}}{s_{k+1}}+(s'_i-s_i)^2)
\end{equation*}
By substituting $s'_i$ with $s_i+(s'_i-s_i)$, we have
\begin{gather}
u_i(c'_i, s'_i)=\frac{s_i c_{k+1}}{s_{k+1}}-c_i+(s_i-s'_i)(s'_i-s_i)^2\nonumber
\end{gather}
which is less than the case where the player acts truthfully.
If the player is not a winner, but wins the game by misreporting, its utility becomes
\begin{align*}
%u_{k+1}(c'_{k+1}, s'_{k+1})&=\frac{s'_{k+1} c_k}{s_k}-c_{k+1}+d(s_{k+1}, s'_{k+1})\\
%&(s_{k+1}-s'_{k+1})\\
u_{k+1}(c'_{k+1}, s'_{k+1})&=\frac{s_{k+1} c_k}{s_k}-c_{k+1}+\left( (s'_{k+1}-s_{k+1})^2\right)\\
&(s_{k+1}-s'_{k+1})
\end{align*}
The first term $\frac{s_{k+1} c_k}{s_k}-c_{k+1}$ is negative, because $\frac{s_k}{c_k}>\frac{s_{k+1}}{c_{k+1}}$. The second term $\left( (s'_{k+1}-s_{k+1})^2\right)(s_{k+1}-s'_{k+1})$ is also negative, $s_{k+1}<s'_{k+1}$. Thus the utility is negative in this case.
\item  \textit{$c' > c$ and $s' > s$.}
If the player is already a winner and still wins by misreporting, then the utility of player decreases. Also, if the player is not a winner, but wins the game by misreporting, then its utility becomes negative.
\end{enumerate}
\end{proof}

We must now prove that the proposed multidimensional auction is individually rational. This property ensures that players participate in the task allocation game.
\begin{lemma}\label{lem4}
The proposed multidimensional auction satisfies the individual rationality constraint (constraint \ref{ir}).\end{lemma}
\begin{proof}
We need to show that when players act truthfully, their utilities are greater than or equal to zero. If the player is not a winner, then its utility is zero. If the player is a winner then its utility is determined as follows.
\begin{align*}
u_{i}(c_{i}, s_{i})=\frac{s_{i} c_{k+1}}{s_{k+1}}-c_{i}
\end{align*} 
Given that $\frac{s_i}{c_i}>\frac{s_{k+1}}{c_{k+1}}$, the utility of the winner is positive.
\end{proof}
\subsection{Task allocation}

We now propose the task allocation policy among the $k$ selected players to reduce the job completion time using different models of compute jobs. 

\textit{Single job}: In some applications such as audio to text conversion, we only have one job. In such a single job scenario, the workload assigned to each player is proportional to its committed service quotient (in the mobility aware approach in the next section, we also consider the auction time interval in the task allocation towards ensuring that the selected nodes are able to complete the assigned task before the distributor preforms the next auction). For example, for player $i$ the amount of workload is equal to $D \frac{s_i}{\sum_{i=1}^k s_i}$, where $D$ is the total workload. In this case, all selected players execute the assigned tasks in parallel that reduces the job completion time. Any other task allocation, reduces the number of tasks executing in parallel and consequently increases the job completion time.

\textit{ DAG jobs}: There are some applications that contain a set of jobs where the execution of jobs have causal ordering. We use a Directed Acyclic Graph (DAG) to represent causal dependencies in a set of jobs. In our DAG, nodes correspond to jobs and directed links represent the causal dependencies. Figure~\ref{fig_dag} shows an example of a DAG representing a set of jobs.

We can use the same task allocation as the single job for DAG jobs. For dependent jobs, we execute jobs sequentially and assign workloads to the selected $k$ players in proportion to their committed service quotients. For independent jobs, we can run them in parallel. For example, job $B$ and $C$ in figure~\ref{fig_dag} are independent and can be run in parallel. Let $D_1$, and $D_2$ be the total workloads of jobs $B$, and $C$, respectively. Then the assigned workload to player $i$ is $D_1 \frac{s_i}{\sum_{i=1}^k s_i}+D_2 \frac{s_i}{\sum_{i=1}^k s_i}$.

\textit{Indivisible jobs}: In some applications, jobs are indivisible. Even for indivisible jobs, the proposed game theoretic framework can be used to provide incentives, utilize nearby resources, and reduce the job completion time.

For a single indivisible job, the distributor selects only one nearby device and assigns the whole job if the job execution time is less than the time taken in executing the job locally. For indivisible DAG jobs, the distributor acts in the same manner as it would for the single indivisible job scenario, for all dependent jobs and executes them sequentially. For all independent jobs, the distributor executes them in parallel. First, it selects $k$ nearby smart phones where $k$ is equal to the number of independent jobs, by considering the budget limit constraint. For task allocation, the distributor ranks the independent jobs in the following manner. For each independent job, the distributor computes the total remaining job load from the node corresponding to the independent job to the leaf of the DAG. Then, it assigns ranks to the independent jobs in decreasing order of the remaining load. Next, the highest rank job is assigned to the selected player with the highest value of committed service quotient. Let $50$, $100$, $100$, $50$ be the workload of jobs $B$, $C$, $D$, and $E$, respectively in Figure~\ref{fig_dag}. Then the total remaining job loads from $B$, and $C$ are $200$, and $150$, respectively. Therefore, in allocating tasks, the distributor assigns $C$ to selected player with the highest committed service quotient.

Although our game theoretic framework can be used for task allocation where jobs are indivisible, it is likely to have lower efficiency compared to local execution, in terms of job completion time. Specially, when we have computationally intensive indivisible jobs in an environment with high mobility (i.e., nearby mobile nodes are disconnected before completing the assigned task). Note that for divisible jobs, we deal with mobility by preforming multiple auctions with dynamic time intervals (see section \ref{secmobility}).  

\begin{figure}[!t]
\centering
\includegraphics[width=2.5in]{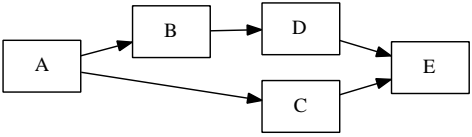}
\caption{Example of DAG jobs.}
\label{fig_dag}
\end{figure}

\section{Mobility aware approach}\label{secmobility}

So far, we have implicitly assumed that the mobile devices are available for entire duration of the task computations. However, because of mobility of the distributor or the players, the topology of network may change and consequently, some mobile smart phones may get disconnected before completing the assigned tasks. Furthermore, some other smart phones might arrive in the vicinity of the distributor with higher computational power and less cost. To deal with such device mobility, instead of holding the auction only once, we hold the auction multiple times. The key challenge in holding the auction multiple times is the determination of the time interval between auctions.

When we perform the auction very frequently, we are able to find newly arriving smart phones. Note that the duration of the task allocated to the smart phones is limited by the time between auctions. Therefore, when the time between auctions is short, the chance that a smart phone gets disconnected before completing its assigned task will decrease. However, the number of auctions and the accompanying overheads will increase, since each time the distributor needs to probe to find the neighboring smart phones and their type values. On the other hand, if we hold the auction less frequently, then there is a higher chance that smart phones get disconnected before performing the assigned tasks. Also, we may lose potential computational power of newly arriving smart phones. As a result, we may see higher job completion times. Therefore, there is a trade-off between the overhead due to the number of times that we perform the auction and the job completion time. 

\begin{figure}[!t]
\centering
\includegraphics[width=3in]{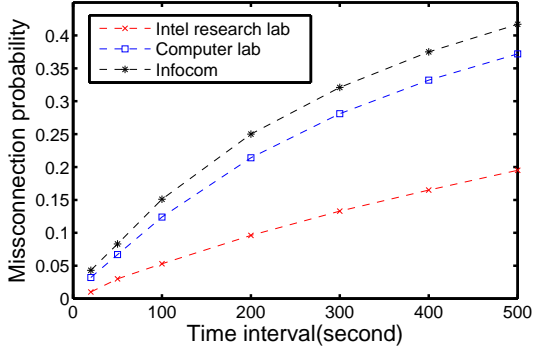}
\caption{ Average percentage of disconnected nodes versus the time interval between auctions.}
\label{fig_missconectionprob}
\end{figure}

Figure \ref{fig_missconectionprob} shows the average percentage of disconnected nodes versus the time interval between two auctions (or distributor probes) for three different real world contact traces. The three real contact traces were collected as a part of the Haggle Project~\cite{cambridge-haggle-2006-01-31} at the Intel research lab, at the University of Cambridge Computer Laboratory, and during the \textit{Infocom05} conference, respectively. This figure shows that in all three experiments, as we increase the time interval from $20$ to $500$ seconds, the average percentage of disconnected nodes increases, consequently, increasing the job  completion time in our task allocation mechanism.

Instead of using a fixed time interval between auctions, we propose an adaptive approach that modifies the time interval dynamically based on the mobility of the nodes in the environment. Our goal is to reduce the number of times we preform the auction without significantly increasing the job completion time.

We use AIMD (Additive Increase and Multiplicative Decrease) approach to adaptively set the time interval between two auctions. We use AIMD for two reasons. First, in real world scenarios, the change in the number of contacts has been shown to follow a power law distribution with a bursty traffic pattern~\cite{slaw}, i.e. a large number of contacts arrive or leave a specific place over a short period of time. This implies that the time interval between two auctions, $T$, should decrease rapidly when a change is observed in the number of connections. Second, AIMD is a stable adaptive approach that is widely used in networking protocols such as TCP. 

We describe our adaptive approach as follows. We start with a fixed time interval, $T$. Then, we update the value of $T$ by comparing the results from the previous auction. If no change in the number of connections is observed, we conservatively increase the value of $T$ linearly (by a constant value of $b$ seconds). Otherwise, we decrease the value of $T$ multiplicatively by a factor of 2. Therefore, the value of $T$ increases slowly when there is no change, and it decreases rapidly in case of change in the number of connections from one auction to the next. More specifically, the value of $T$ is obtained from the following formula.
\begin{equation}\label{eqadaptiveinc}
T(i+1)=\left\{
\begin{array}{@{} l c @{}}
\max (\frac{T(i)}{2}, T_{min}) &Nc_{i}\neq Nc_{i+1}\\
\min(T(i)+b, T_{max}) &\text{otherwise}
\end{array}\right.
\end{equation}
where, $T_{min}$ and $T_{max}$ are the minimum and the maximum values of $T$, respectively. $Nc_{i}$ denotes the contact set at the beginning of time interval $i$. 

\section{Privacy}\label{secprivacy}
Another important issue in MCC is the privacy of data. In our game theoretic framework the distributor violates the privacy of data by offloading a computational job to the nearby mobile devices. Let $j=\{1, \ldots, m\}$ denote a set of computational jobs that the distributor wants to execute. We formulate the privacy risk for distributor, $PR(d)$, as the following:
\begin{equation}\label{equprivacyrisk}
PR(d)=\sum_{j=1}^m \rho_j v_j
\end{equation}
$\rho_j\in [0, 1]$ in equation~\ref{equprivacyrisk} determines the sensitivity of the computational job $j$. $v_j$ represents the visibility of job $j$ and is equal to the number of selected mobile devices by the distributor for executing tasks of job $j$.

As the above formulation shows, the privacy risk for the distributor, $PR(d)$, increases by increasing the number of selected mobile devices. However, the job completion time decreases as the number of selected mobile devices is increased. Figure~\ref{figprivacydis} shows the tradeoff between $PR(d)$ and the job completion time. Note that the distributor can provide privacy by adding another constraint for privacy risk in both allocation and payment mechanism at the cost of increasing the job completion time.

Besides distributor, the nearby mobile devices also violate their privacy by revealing their computational capabilities and the cost of computation. To examine the impact of privacy on the nearby mobile users utilities, we define the privacy parameter $x\in [0, 1]$ where the mobile device conceals the actual value of cost, $c$, and the actual value of service quotient, $s$, by a factor $x$. I.e, $c'=x c$ and $s'=\frac{1}{x} s$ (we only consider the case where the mobile device decreases the actual value of cost and increases the actual value of service quotient. For all the other cases, it is obvious that the utility of mobile device decreases by concealing the actual values of cost and service quotient). Figure~\ref{figprivacymobile} shows the CDF of nearby mobile devices' utilities for different values of $x$. This figure shows that the mobile devices can provide privacy at the cost of decreasing their utilities and the highest utility is obtained if the mobile devices act truthfully, $x=1$, and reveal their actual costs and computational capabilities as we proved in Section~\ref{subsecproofs}. 
\begin{figure*}[!t]
\centerline{
\subfigure[]{\label{figprivacydis}\includegraphics[width=3in]{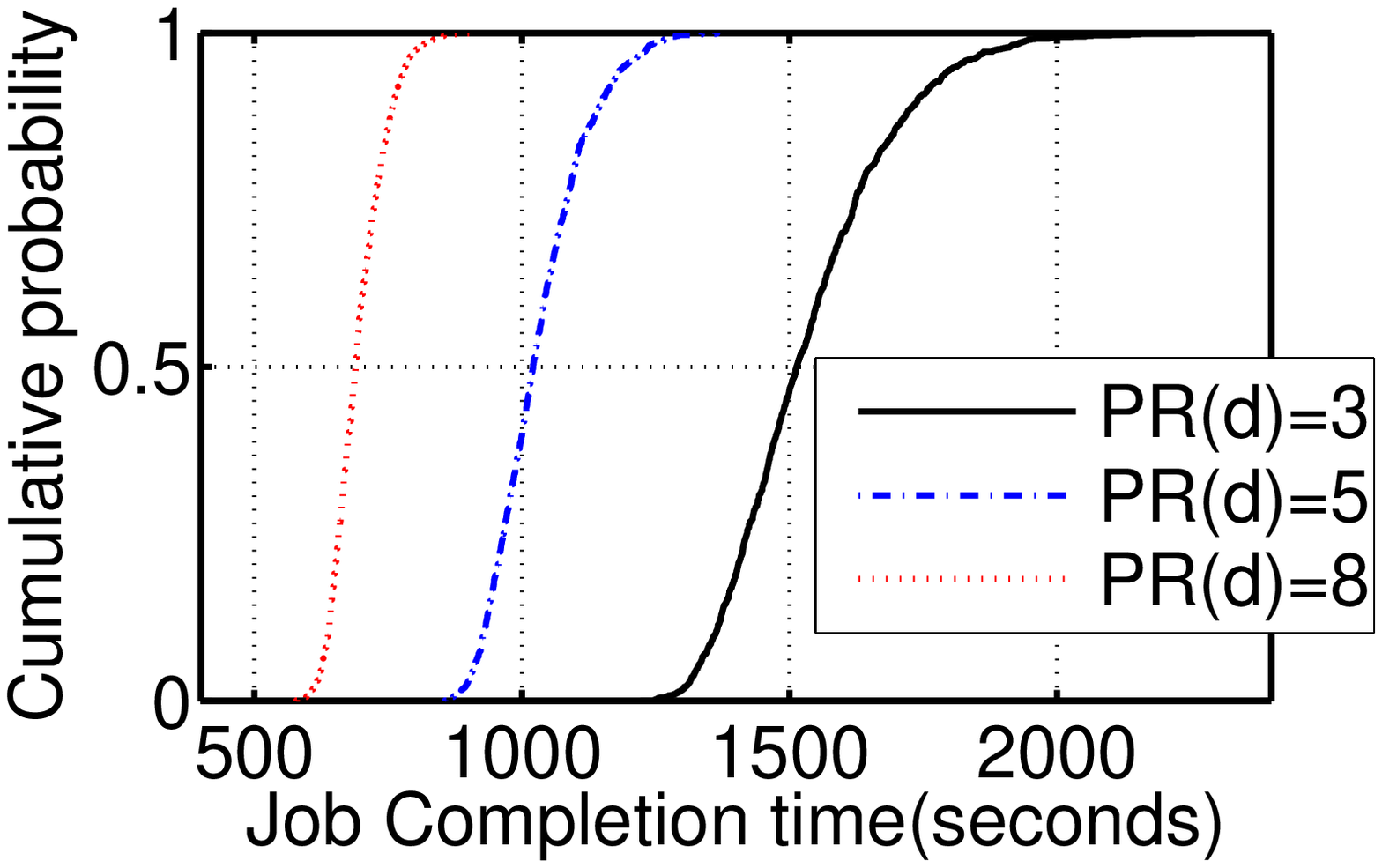}}
\hfil
\subfigure[]{\label{figprivacymobile}\includegraphics[width=3in]{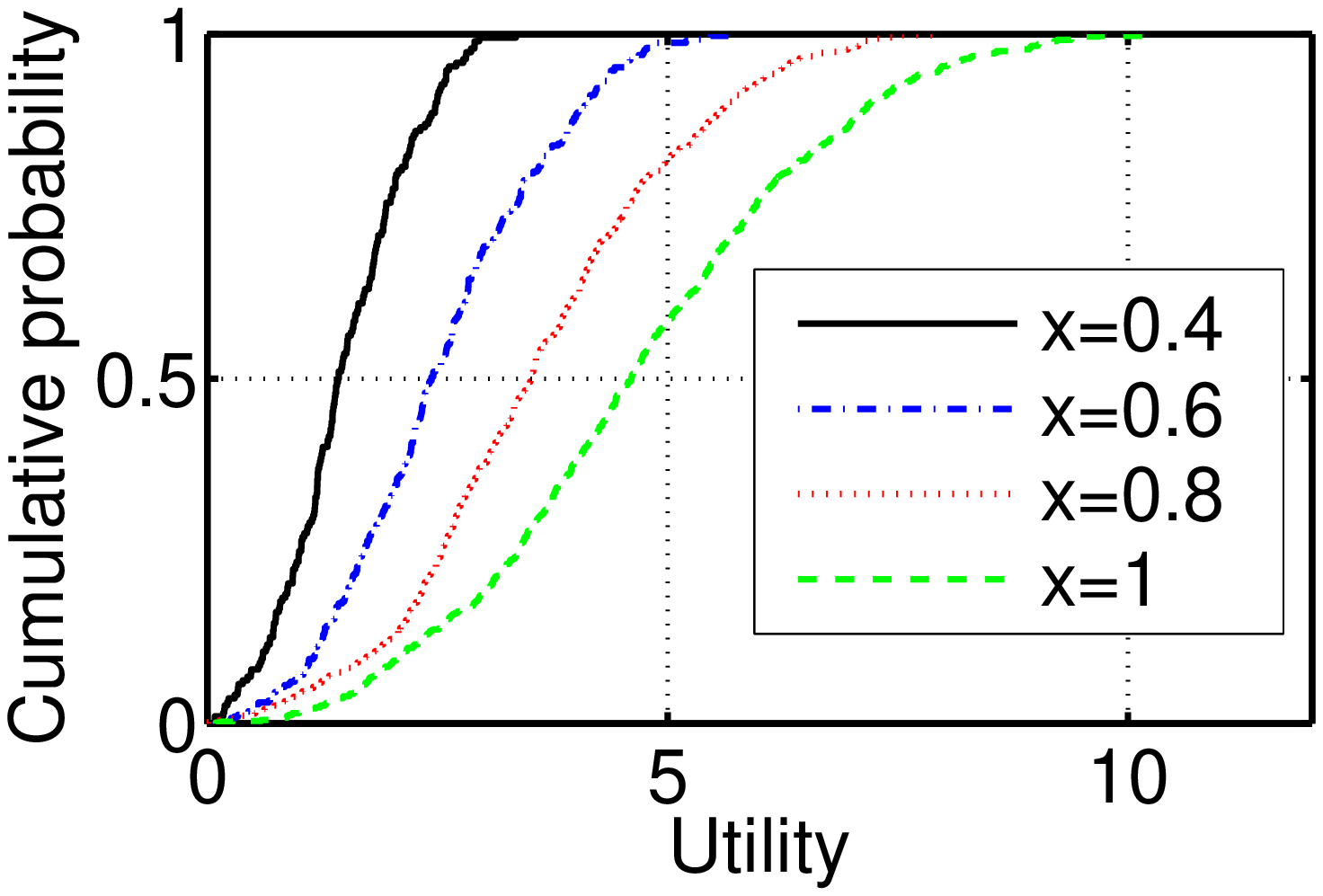}}}
\caption{(a) CDF of job completion time for different values $PR(d)$ (b) CDF of utility of mobile devices for different values of privacy parameter, $x$.}
\label{fig_privacy}
\end{figure*}
\section{Evaluation}\label{secevaluation}

\begin{table*}[!t]
%% increase table row spacing, adjust to taste
%\renewcommand{\arraystretch}{1.3}
% if using array.sty, it might be a good idea to tweak the value of
% \extrarowheight as needed to properly center the text within the cells\\
\caption{Comparison of the mobility aware approach with fixed and adaptive time intervals between auctions.}
\label{table_compare}
\centering
%% Some packages, such as MDW tools, offer better commands for making tables
%% than the plain LaTeX2e tabular which is used here.
\resizebox{17.5cm}{!}
{\begin{tabular}{c|c c c|c c c|c c c}
\hline\hline
&\multicolumn{3}{ |c| }{Intel research lab} & \multicolumn{3}{ |c| }{Computer lab} & \multicolumn{3}{ |c }{Infocom}\\
\hline
Approaches & Completion& Number of & Percentage of &Completion& Number of& Percentage of & Completion & Number of & Percentage of \\
& time(second)& auctions& disconnected nodes & time(second)& auctions& disconnected nodes & time(second)& auctions& disconnected nodes \\
\hline
$T=20s$ & 2109 & 119 & 0.04 & 2100 & 119 & 0.054 & 2632 & 153 & 0.066\\
\hline
$T=200s$ & 2646 & 15 & 0.247 & 3021 & 18 & 0.4 & 3776 & 23 & 0.43\\
\hline
$T=500s$ & 3288 & 7 & 0.267 & 3718 & 10 & 0.5 & 4725 & 12 &0.57\\
\hline
Heuristic approach & 2337 & 19 & 0.16 & 2294 & 33 & 0.16 & 2997 & 43 &0.2\\
\hline
\end{tabular}}
\end{table*}

\begin{figure*}[!t]
\centerline{
\subfigure[]{\label{fig_100MB}\includegraphics[width=2.5in]{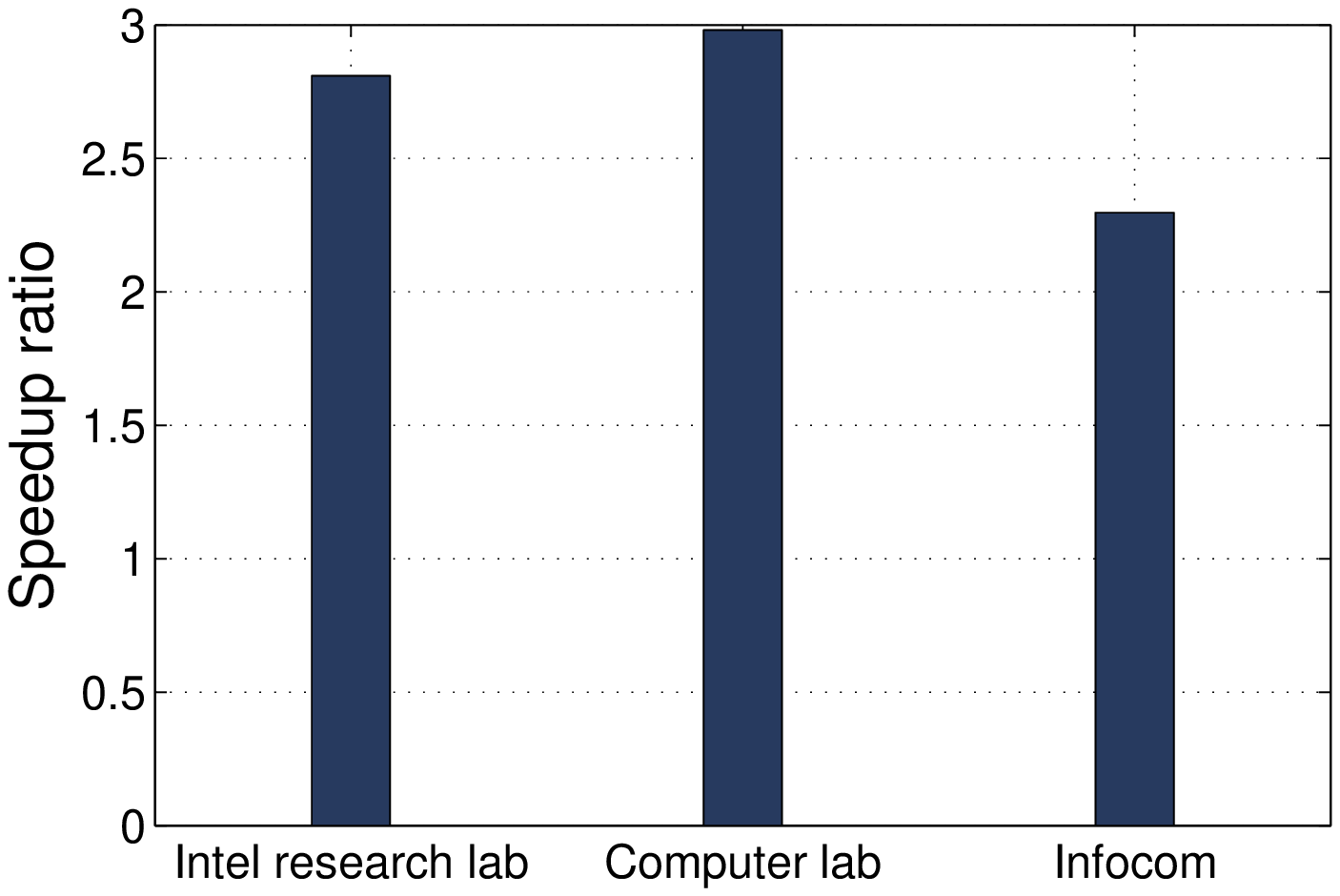}}
\hfil
\subfigure[]{\label{fig_300MB}\includegraphics[width=2.5in]{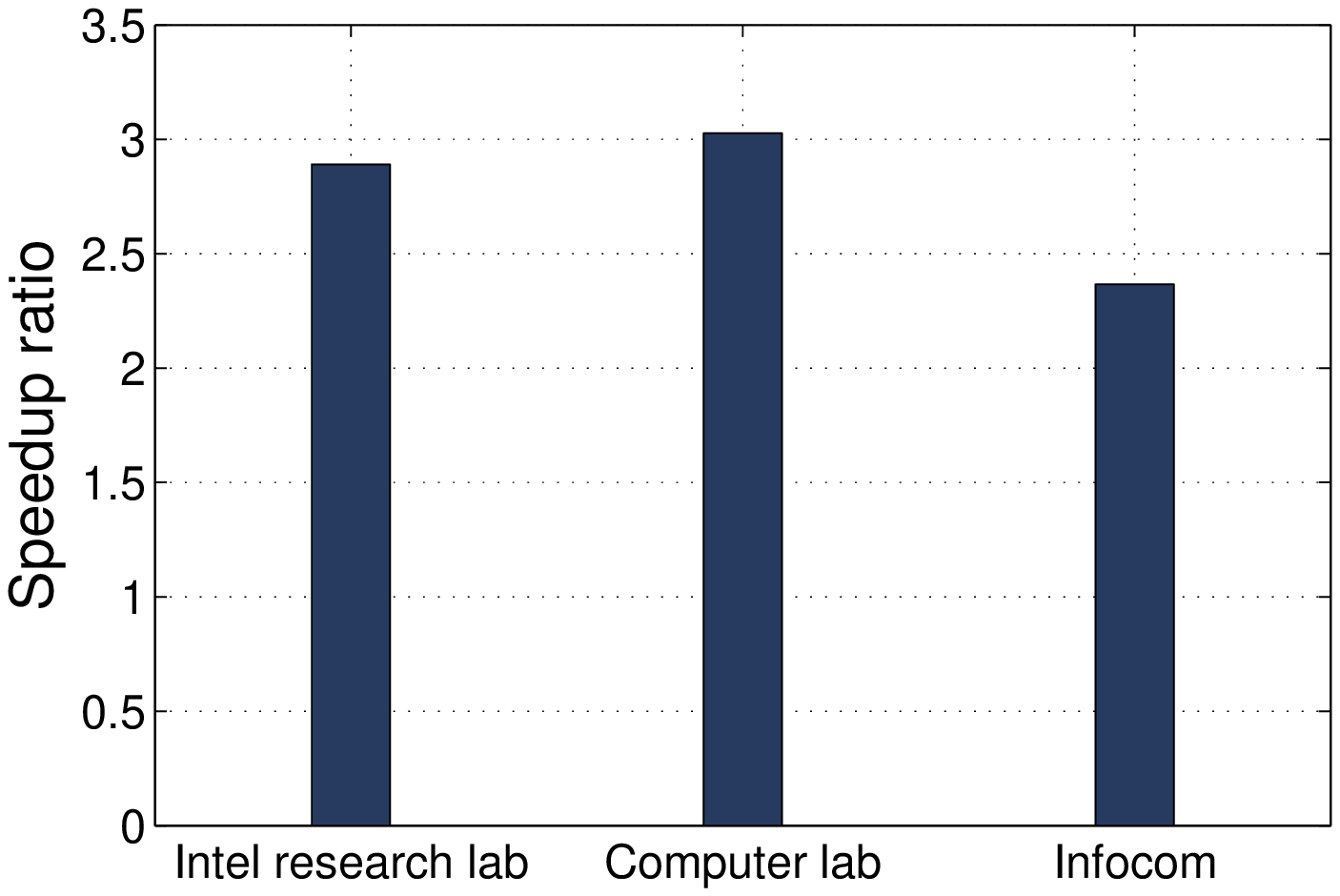}}
\hfil
\subfigure[]{\label{fig_600MB}\includegraphics[width=2.5in]{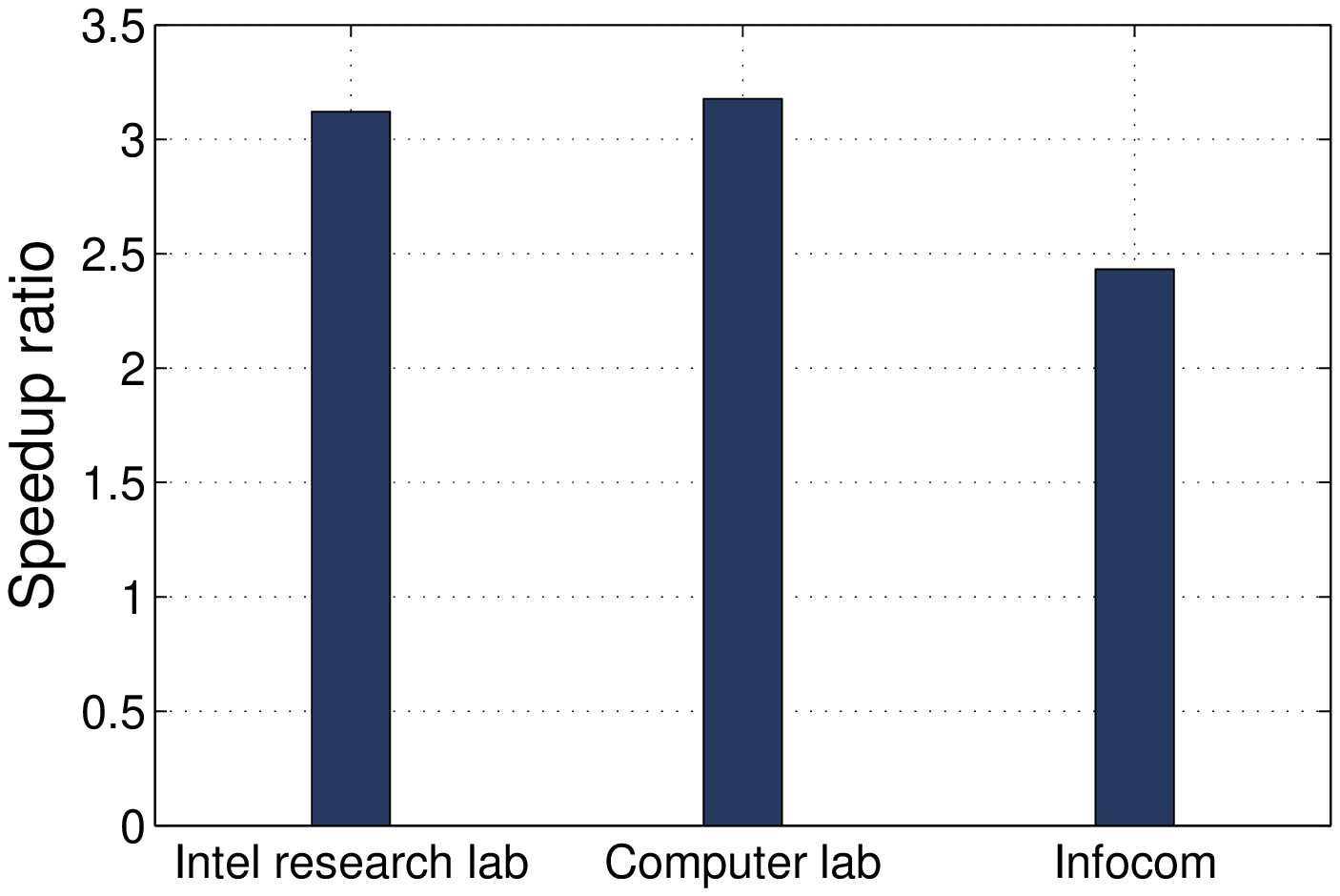}}}
\caption{Speedup ratio of mobility aware approach using our heuristic to adaptively select the time interval between auctions for different file sizes, (a) $100MB$, (b) $300MB$, (c) $600MB$.}
\label{fig_barmobilityaware}
\end{figure*}
\subsection{Experimental Setup}
To evaluate our task allocation method and its overhead, we use a custom simulator, that we develop using Matlab, to simulate the speech to text conversion application under different mobility conditions. Instead of performing the conversion at the distributor, the audio file of the speech is divided into smaller files and nearby smart phones are used to convert the small audio files to text files and send back the results to the distributor who finally combines them. 

%In this setting, when a smart phone, distributor, wants to distribute a task among its neighbors, it first sends a probe message to find the nearby smart phones, and the neighboring smart phones declares the cost of converting the audio file to the text file per second, $c$, and the amount of audio file that they can
%convert to the text file per second, $q$. Based on the received information, the distributor selects a set of nearby smart phones and divides the audio file in proportion to the selected players $q$s' and assign them to the selected players. 
In our simulation, each player chooses its cost of processing the audio file ($c$) from a uniform distribution in interval $[1, 5]$. The size of of audio file that a player can process per second ($s$) is selected randomly from a uniform distribution $[50KB, 125KB]$. The distributor's budget limit function ($B(u_d)$) is a constant function $50$ per megabyte of data for the distributor utility $u_d$ to be positive (note that other functions are also possible). Also, $T_{min}$, $T_{max}$, that represent the maximum and minimum values for the time interval between auctions, are set to $20s$, and $500s$ respectively. Finally, the linear increase factor in the heuristic approach ($b$) is set to $20s$ in all experiments. All the results are obtained by averaging over $1000$ randomly selected samples.  

We evaluate our task allocation method for two scenarios. In our first scenario, we assume that the neighboring nodes are available during the whole computation. In our second scenario, we consider the mobility of nodes and the possibility of disconnection during computations. To evaluate the mobility aware approach in the second scenario, we use three real contact traces with different contact properties. These traces were collected as a part of the Haggle Project~\cite{cambridge-haggle-2006-01-31}. The first data set was collected at Intel research lab, the second set was collected at the University of Cambridge Computer Laboratory, and the third set was collected during the Infocom05 conference. %These traces provide different contact properties. 
We select $9$ nodes from the first set and $12$ nodes from the second and third sets.
We also use two synthetic mobility models namely the Slaw mobility model~\cite{slaw} and the Random Walk mobility model(RW)~\cite{khaledimobility2}. The Slaw model captures many of the statistical properties of human mobility, while the Random Walk model provides high mobility and abrupt changes in movement patterns~\cite{khaledimobility2}. 

\subsection{Results}
We evaluate our task allocation method in terms of both the job completion time and the overhead. To quantify the overhead, we measure the number of times that the auction is preformed before task completion. We also use the speedup ratio metric to evaluate the speed up in the job completion time as a result of offloading tasks to other smart phones. The speedup ratio is obtained by dividing the job completion time in locally executing the job by the job completion time in our distributed approach. The speedup ratio shows the performance benefits where we provide the required incentives for the nearby mobile devices to lend their resources with the case where there is no incentive model and selfish nearby mobile devices would not be willing to lend their resources. In the other words, the speedup ratio demonstrates the performance benefits of the distributor from using the game theoretic framework.

Figure \ref{fig_static_bar} shows the speedup ratio for the first case where selected neighboring nodes are available during the whole computation. We use three audio files with sizes of $100MB$, $300MB$, and $600MB$. The number of neighboring nodes are randomly selected from a power law distribution. Figure \ref{fig_static_bar} shows how our method significantly improves the job completion time in comparison to the time taken in locally executing the job. This figure shows that for all of our chosen file sizes, our method can improve the job completion time by about a factor of $5$. %From $1222$, $3654$, and $7290$ seconds to $224$, $687$, and $1430$ seconds for file size of $100MB$, $300MB$, and $600MB$ respectively. %These number shows that file size dose not affect the speed up ratio too much.
  
Table \ref{table_compare} shows the job completion time, number of preformed auctions, and the percentage of disconnected nodes in case of mobility for the three real contact traces. In this table, we compare the results of fixed time intervals with our heuristic approach for adaptively selecting the time intervals between auctions. In this experiment, the file size is set to $600MB$. We make the following observations from Table \ref{table_compare}. First, as the time interval between two auctions ($T$) increases or the number of performed auctions decreases, the job completion time and the percentage of disconnected nodes increase. By increasing the time interval, the percentage of disconnected nodes increases and the distributor also misses some of the newly arriving nodes. This reduces the number of tasks that are being executed in parallel, and consequently the job completion time increases. Second, the fixed time interval approach performs well in terms of job completion time for small values of $T$. For large values of $T$, this approach preforms well in terms of overhead, i.e., the number of performed auctions. Thus, finding the right value for $T$ where it preforms well in terms of both the job completion time and the overhead depends on the mobility patterns of mobile nodes. Third, the heuristic approach preforms well in terms of both the job completion time and the overhead without requiring any prior knowledge of the mobility patterns of mobile nodes.

Figure \ref{fig_compare}, compares the performance of our heuristic approach for finding time intervals between auctions with the fixed time intervals approach, when nodes are mobile. In this figure, the x-axis represents the normalized job completion time, and the y-axis represents the normalized number of preformed auctions. Our goal is to reach to the left bottom corner of the figure where minimum job completion time and minimum overhead exist. Figure \ref{fig_compare} shows that the heuristic approach outperforms the fixed intervals method and gets close to the left bottom corner.

Next, we compare the job completion time of the mobility aware approach with that of local execution using the three real world contact traces for different file sizes. Figure~\ref{fig_barmobilityaware} shows the speedup ratio of the mobility aware approach using our heuristic approach. This figure shows that the speedup ratio is close to $3$. However, the amount of improvement in the Infocom05 contact trace is a less than the other contact traces.  This is because the average percentage of disconnected nodes in the Infocom05 trace is larger than the other traces (see Figure \ref{fig_missconectionprob}).

To further analyze the effect of contact traces, we use two synthetic mobility models to generate the contact traces. Specifically, we use the Slaw~\cite{slaw} and Random Walk mobility models. In our setting, we assume that 20 mobile nodes move in an area of $1000m \times 1000m$ for one hour and each mobile node has a circular transmission range of $150m$. Also, the audio file size is set to $600MB$. 

Figure~\ref{fig_rwslaw} shows the speedup ratio of mobility aware approach with fixed and heuristic approaches for selecting the time interval between auctions in the Slaw and Random Walk models. This figure shows that our proposed model can improve the job completion time by a factor of $5$ to $7$ in the Slaw model (that mimics the human walk patterns). However, the speedup ratio of the mobility aware approach is at most $3$ in the Random Walk model. This is because, in the Random Walk model, mobile nodes continuously move from one location to the another in a random fashion. According to our measurements in the Random Walk model, more than half of the selected neighboring nodes get disconnected before completing their tasks, consequently increasing the job completion time. However, we can still improve the job completion time by a factor of $3$ for small fixed time interval ($T=20s$) at the cost of increased auction overhead. Furthermore, we can improve the job completion time by factor $2$ using our heuristic approach without significantly increasing the communication overhead. According to our experiments, the number of preformed auctions using the heuristic method and using a small fixed time interval ($T=20s$), are $59$ and $110$, respectively. It should be noted that for a high mobility contact trace (such as Random Walk), there is a higher change of contact among mobile nodes. Therefore, we can still improve the job completion time. Also, it is worth mentioning that by reducing the time interval between auctions or equivalently reducing the workloads of the assigned tasks, we can further improve the speedup ratio, but at the cost of increased overhead.

To understand the impact of the number of nearby nodes, we vary the number of mobile nodes between $10$ and $50$ with increments of $10$, in the Slaw model. Figure \ref{fig_slaw_exetime} shows the job completion time of the proposed approach using the Slaw mobility model as a contact trace. This figure shows that as the number of nodes increases, the job completion time decreases by more than $40\%$, when fixed time intervals ($T$) are used. Job completion time reduces from $1486$ to $833$ seconds for $T=20$ seconds, from $1888$ to $1055$ seconds for $T=200$ seconds, and from $2238$ to $1324$ for $T=500s$. In addition, this figure shows that our heuristic approach preforms well in terms of job completion time. %Moreover, we can see that the amount of reduction in the execution time in the heuristic approach by increasing the number of nodes is less than the reduction in the fixed time interval.  
 
Figure \ref{fig_compare_slaw} shows both the job completion time and the overhead (number of preformed auctions) using the Slaw mobility model with different number of nodes. This figure shows that for different number of nodes in the Slaw mobility model, the heuristic approach preforms significantly better than the fixed time intervals method, in terms of both the job completion time and the overhead. 

Finally, we evaluate our methods for DAG jobs. The structure of jobs is the same as Figure~\ref{fig_dag} and the workloads of jobs $A$, $B$, $C$, $D$, and $E$ are $100$, $50$, $100$, $100$, $50$ respectively. Figure~\ref{fig_dagbar} shows the speedup ratio of the mobility aware approach for DAG jobs, using the heuristic method to select the time intervals. This figure shows that our proposed approach improves the job completion time by factor $2.5$ to $3$ for DAG jobs. DAG jobs have the same behaviour as the single job by varying the number of nodes or using different mobility models. 

\begin{figure}[!t]
\centering
\includegraphics[width=3in]{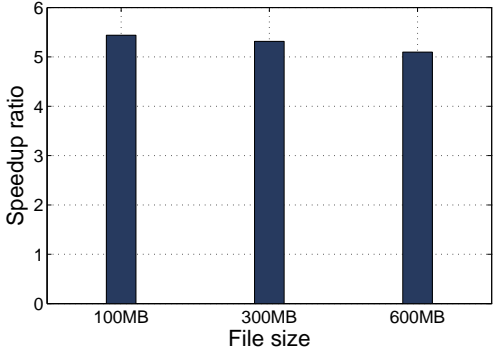}
\caption{Speedup ratio without considering mobility for different file sizes.}
\label{fig_static_bar}
\end{figure}

\begin{figure}[!t]
\centering
\includegraphics[width=3in]{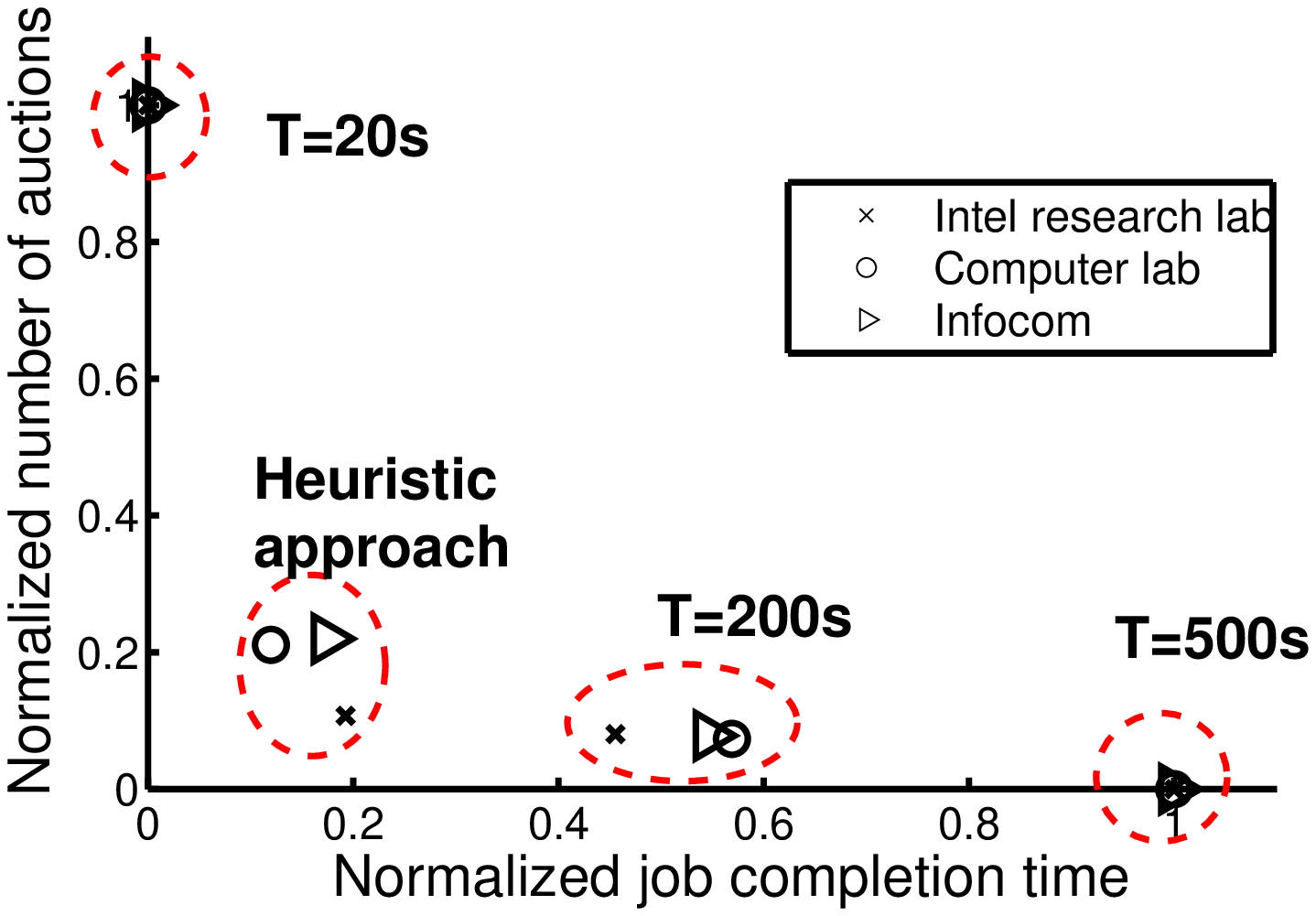}
\caption{Comparison of approaches with fixed and heuristic-based time intervals between auctions, in the presence of mobility.}
\label{fig_compare}
\end{figure}

\begin{figure}[!t]
\centering
\includegraphics[width=3in]{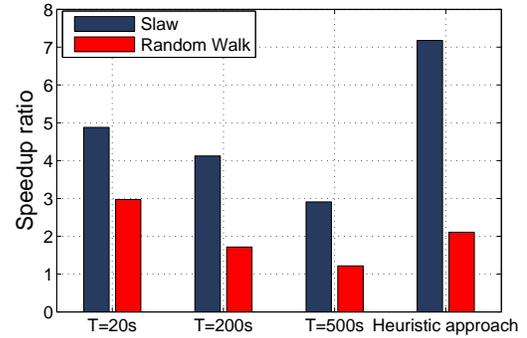}
\caption{Speedup ratio of mobility aware approach using fixed and heuristic-based time intervals between auctions in Slaw and Random Walk models.}
\label{fig_rwslaw}
\end{figure}

\begin{figure}[!t]
\centering
\includegraphics[width=3in]{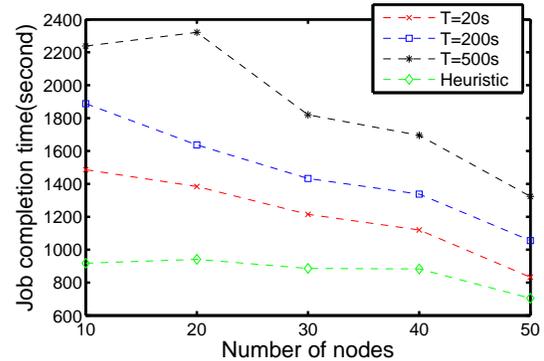}
\caption{The job completion times for different values of $T$ and the heuristic approach versus the number of nodes in the Slaw mobility model.}
\label{fig_slaw_exetime}
\end{figure}
\begin{figure}[!t]
\centering
\includegraphics[width=3in]{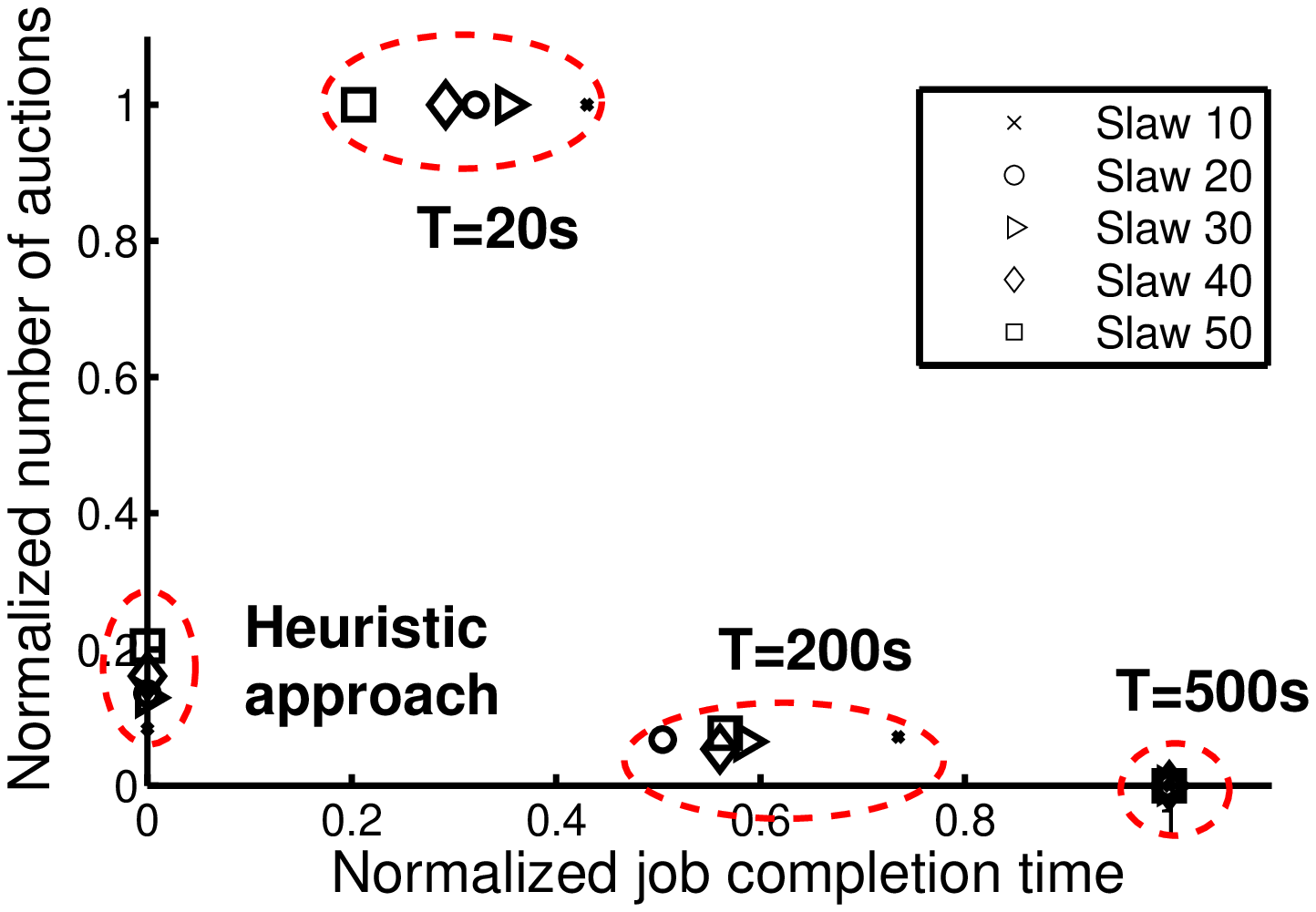}
\caption{Comparison of approaches with fixed and heuristic-based time intervals between auctions in the presence of mobility based on the Slaw mobility traces.}
\label{fig_compare_slaw}
\end{figure}

\begin{figure}[!t]
\centering
\includegraphics[width=3in]{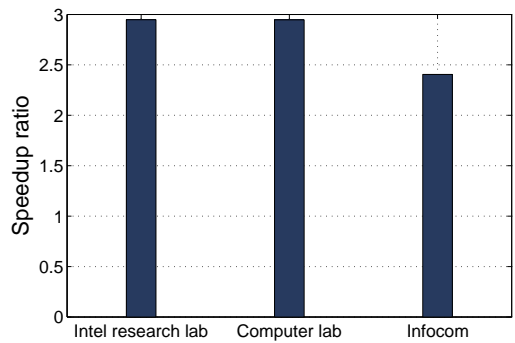}
\caption{Speedup ratio of mobility aware approach using our heuristic for selecting the time interval between auctions in DAG jobs for three real contact traces.}
\label{fig_dagbar}
\end{figure}

\section{Conclusion}\label{secconclu}
We presented and evaluated a game theoretic framework for task allocation in mobile cloud computing environments comprising of selfish mobile devices. Specifically, we proposed a multidimensional auction for allocating the tasks of a job among nearby mobile nodes based on their computational capabilities and also the cost of computation at these nodes with the goal of reducing the overall job completion time and be beneficial to all the parties involved. We considered node and task heterogeneity as well as node mobility in developing our methods. Our evaluations demonstrated the benefits of our methods.
\end{document}